\documentclass[12pt]{iopart}
\pdfoutput=1
\usepackage[colorlinks=true,linkcolor=magenta,urlcolor=blue,citecolor=blue]{hyperref}
\expandafter\let\csname equation*\endcsname\relax
\expandafter\let\csname endequation*\endcsname\relax
\usepackage{amsmath,amsthm,amssymb,amsfonts,graphicx,tabularx,dcolumn,bm,xcolor,float,algorithm,algorithmic,tikz}
\usepackage{iopams}

\usepackage{braket}
\usepackage{verbatim}
\usepackage{Qcircuit}
\usepackage{color}
\usepackage{epstopdf}

\newcommand{\cC}{\mathcal{C}}
\newcommand{\cX}{\mathcal{X}}
\newcommand{\cY}{\mathcal{Y}}
\newcommand{\cL}{\mathcal{L}}
\newcommand{\fc}{f_{C}}
\newcommand{\rhat}{\hat{R}}

\newcommand{\cF}{\mathcal{F}}
\newcommand{\erm}{\hat{h}}
\newcommand{\cS}{S}
\newcommand{\fh}{f_h}
\newcommand{\tfh}{\tilde{f_h}}
\newcommand{\layer}{B_nB_{n-1}\cdots B_2}
\newcommand{\CNOT}{\text{CNOT}}
\newcommand\myeqi{\mathrel{\stackrel{\makebox[0pt]{\mbox{\normalfont\tiny (i)}}}{=}}}
\newcommand\myeqii{\mathrel{\stackrel{\makebox[0pt]{\mbox{\normalfont\tiny (ii)}}}{=}}}
\newcommand\myeqiii{\mathrel{\stackrel{\makebox[0pt]{\mbox{\normalfont\tiny (iii)}}}{=}}}

\newcommand\myeqvi{\mathrel{\stackrel{\makebox[0pt]{\mbox{\normalfont\tiny (vi)}}}{=}}}

\newcommand\myineqiv{\mathrel{\stackrel{\makebox[0pt]{\mbox{\normalfont\tiny (iv)}}}{\leq}}}
\newcommand\myineqv{\mathrel{\stackrel{\makebox[0pt]{\mbox{\normalfont\tiny (v)}}}{\leq}}}

\newcommand\myineqvii{\mathrel{\stackrel{\makebox[0pt]{\mbox{\normalfont\tiny (vii)}}}{\leq}}}

\usepackage{thmtools}
\usepackage{thm-restate}

\newtheorem{theorem}{Theorem}
\newtheorem{lemma}{Lemma}
\newtheorem{definition}{Definition}
\newtheorem{proposition}{Proposition}

\newcommand{\bk}[1]{\textcolor{black}{#1}}

\begin{document}
\title{Sample Complexity of Learning Parametric Quantum Circuits}
\author{Haoyuan Cai}
\address{Center for Quantum Information, IIIS, Tsinghua University, Beijing
100084, People's Republic of China}
\author{Qi Ye}
\address{Center for Quantum Information, IIIS, Tsinghua University, Beijing
100084, People's Republic of China}

\author{Dong-Ling Deng}
\ead{dldeng@tsinghua.edu.cn}
\address{Center for Quantum Information, IIIS, Tsinghua University, Beijing
100084, People's Republic of China}
\address{Shanghai Qi Zhi Institute, 41th Floor, AI Tower, No. 701 Yunjin Road, Xuhui District, Shanghai 200232, China}

\begin{abstract}
Quantum computers hold unprecedented potentials for  machine learning applications. Here, we prove that physical quantum circuits are PAC (probably approximately correct) learnable on a quantum computer via empirical risk minimization: to learn a parametric quantum circuit with at most $n^c$ gates and each gate acting on a constant number of qubits, the sample complexity is bounded by $\tilde{O}(n^{c+1})$. In particular, we explicitly construct a family of variational quantum circuits with $O(n^{c+1})$ elementary gates arranged in a fixed pattern, which can represent all physical quantum circuits consisting of at most $n^c$ elementary gates. Our results provide a valuable guide for quantum machine learning in both theory and practice.
\end{abstract}
\maketitle
\section{Introduction}

Over the past few decades, machine learning, especially deep learning,  has made dramatic progress \cite{Jordan2015Machine,lecun2015deep} in a wide range of tasks, such as playing the game of Go \cite{silver2016mastering,Silver2017Mastering}, protein structure prediction \cite{senior2020improved}, and  computer vision \cite{krizhevsky2012imagenet}, etc. 
More recently, the interplay between machine learning and quantum physics has attracted tremendous interest \cite{Sarma2019Machine,Biamonte2017Quantum,Ciliberto2017Quantum, Dunjko2018Machine,Carleo2019Machine}, giving birth to an emergent research frontier of quantum machine learning. A number of notable quantum algorithms, such as the Harrow-Hassidim-Lloyd (HHL) algorithm \cite{Harrow2009Quantum}, quantum generative models \cite{gao2018quantum}, and quantum support vector machine \cite{Rebentrost2014Quantum}, have been designed to enhance, speed up, or innovate machine learning with quantum devices. These algorithms bear the intriguing potentials of exhibiting exponential advantages compared to their classical counterparts, although subtle caveats do exist and require careful examinations in practice \cite{Aaronson2015Read}.

In 1984, Valiant introduced the  PAC  learning model \cite{valiant1984theory}, which gives a complexity-theoretical foundation and a mathematically rigorous framework for studying machine learning. Since then, the PAC learning model has been  extensively studied in various machine learning scenarios  to understand why and when efficient learning is possible or not \cite{Haussler1990Probably,Shalev-Shwartz2014Understanding}. With the rapid progress in quantum computing \cite{Arute2019Quantum,zhong2020quantum,wu2021strong}, practical applications of quantum machine learning have become more and more realistic \cite{Schuld2017Implementing,Schuld2020Circuitcentric,Beer2020Training,Cong2019Quantum,Watts2019Exponential}.  A natural problem is then to generalize the PAC learning model to quantum learning scenarios. Indeed, notable progress has been made along this direction \cite{Arunachalam2017Optimal,Chung2021Sample,Arunachalam2017Survey,Arunachalam2020Quantum,Heidari2021Theoretical,Sweke2021Quantum,bu2021statistical,caro2020pseudo,du2021efficient,caro2021encoding,Cheng2015Learnability}. For example, in Ref. \cite{Chung2021Sample} Chung and Lin have studied the sample complexity of learning  quantum channels and demonstrated that we can PAC-learn a polynomial-size quantum circuit with a polynomial number of samples. In addition, in Ref. \cite{bu2021statistical} Bu \textit{et al}. investigated the Rademacher complexity of quantum circuits in the framework of quantum resource theories \cite{Chitambar2019Quantum}. They introduced a resource measure of magic for quantum channels based on the $(p,q)$ group norm and found useful bounds for how the statistical complexity scales with resources in the quantum circuits. Yet, this fledgling research direction is still in its rapidly growing early phase and many important issues remain barely explored.

In this paper, we study the problem of the sample complexity for learning parametric quantum circuits. We focus on the supervised learning scenarios and prove that all the unitary physical quantum circuits are PAC learnable on a quantum computer via empirical risk minimization. More concretely, we prove the following two theorems: 1) any physical $n$-qubit quantum circuit consisting of at most $n^c$ unitary gates with each gate acting on a constant number of qubits can be represented in an exact fashion by a family of variational quantum circuits with $O(n^{c+1})$ elementary gates arranged in a fixed uniform pattern; 2) this family of variational quantum circuits is PAC learnable. Since most quantum circuits that can be efficiently implemented on a quantum computer,  such as the circuits for the Shor's algorithm \cite{shor1994algorithms} or the HHL algorithm \cite{Harrow2009Quantum}, contain at most a polynomial number of gates, our results imply that they are all PAC learnable with a quantum computer.

\begin{figure*}
\hspace*{0.1\textwidth}
\includegraphics[width=0.8\linewidth] {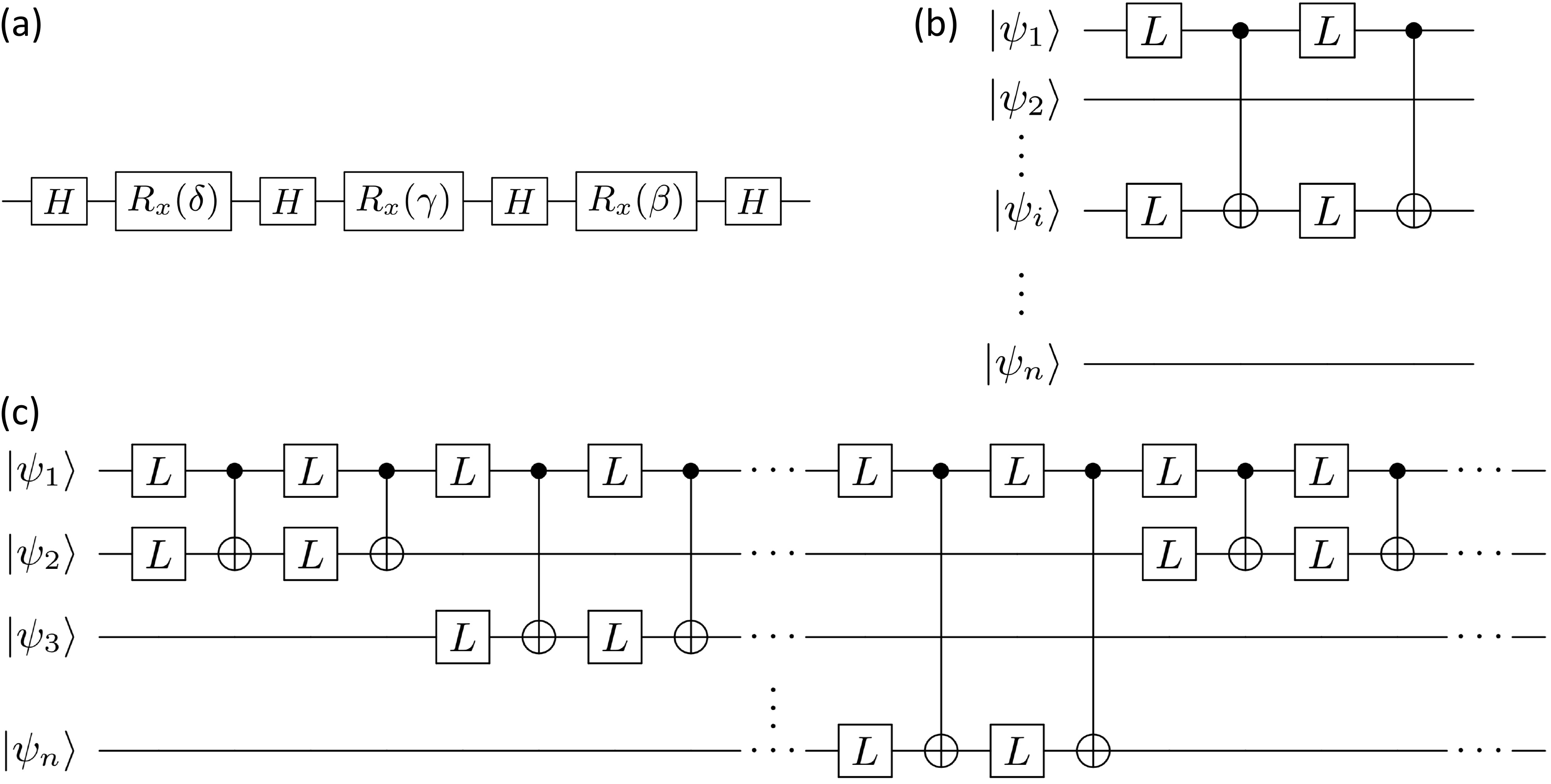}
\caption{Pictorial illustration for the construction of the hypothesis quantum circuits. (a) The elementary (level-$1$) block $L$. This block contains four  Hadamard gates and three single-qubit rotations along the $x$ direction with rotation angles parameterized by $\delta$, $\gamma$, and $\beta$, respectively. (b) The level-$2$ block $B_i$ constructed based on four level-$1$ block and two CNOT gates with the first qubit being the control qubit and the $i$-th qubit being the target one. (c) The constructed hypothesis quantum circuit, which consists of $M n^c$ repeated layers of $B_n B_{n-1}\cdots B_2$.}
\label{VariationF}
\end{figure*}

\section{Results}

\subsection{Notations and the general setting}
We define the concept class $\cC$ as the collection of all the $n$-qubit parametric quantum circuits with at most $n^c$ unitary gates, each gate acting on at most $b$ qubits ($b,c$ are constant numbers independent of $n$). We note that $\cC$ is general enough to include most quantum circuits in practical applications. Here, we study the learnability of the  quantum circuits in $\cC$ under the PAC learning framework \cite{Shalev-Shwartz2014Understanding}. Let $C \in \cC$  be any $n$-qubit circuit in this concept class. When we input an $n$-qubit pure state $\ket{\psi}_{\text{in}}$ to $C$, we will get an output $n$-qubit pure state $\ket{\psi}_{\text{out}} = C\ket{\psi}_{\text{in}}$. 
Therefore, $C$ can be viewed as a function $ \fc :  \cX  \rightarrow \cY$,
where its domain $\cX$ and range $\cY$ are both the set of all $n$-qubit pure states. In this work, we write $x \in 
\cX$ as an abbreviation of the $n$-qubit quantum state $\ket{\psi(x)}$, and similarly for $y \in \cY$. With these notations, we sometimes write $y=\fc(x)$ to denote $\ket{\psi(y)} = C\ket{\psi(x)}$ for simplicity.

We consider the supervised learning scenario \cite{Shalev-Shwartz2014Understanding} and  denote the training set of size $m$ as $\cS = \{(x_1,y_1),(x_2,y_2),\cdots,(x_m,y_m)\}$. Under the PAC learning framework, to learn the unknown circuit $C$, assume we have $m$ independent $n$-qubit input samples $\{x_1,x_2,\cdots,x_m\} \in \cX^m$, we can input them into $C$, and obtain output states $\{y_1,y_2,\cdots,y_m\} \in \cY^m$. 
The essential task of supervised learning is then to learn from $\cS$ a hypothesis function (here a quantum circuit $F$) that can approximate the target function $\fc(x)$. This might be accomplished by minimizing certain loss functions over a set of variational model parameters. More concretely, we construct a variational quantum circuit $F$ consisting of multiple gates with some of them having tunable parameters. By tuning these parameters, we can use $F$ to represent different functions $f_h: \cX \rightarrow \cY$, and we define our hypothesis space $\cF$ as the collection of all the functions $f_h$ that $F$ can represent. 
Given $m$ independent samples and a tunable quantum circuit $F$, we can use the following process to make our circuit $F$ become a good approximation of $C$. We can tune the parameters of $F$ according to the training set $\cS$, so that when we put the state $x_i$ ($i = 1,2,\cdots,m$) into the input of $F$, the output of $F$ will be a good approximation of $y_i$. By the PAC learning theory, assuming that $\cF$ has good generalization power, $f_h$'s decent performance on the training set can imply its good performance over the whole sample space.

The effectiveness of the above process is based on two assumptions. First, the space of $\cF$ should be large enough, so that given any set of samples $\cS=(x_i,y_i)_{i=1,2,\cdots,m}$, we can always find a function $f_h \in \cF$, such that $f_h(x_i)$ can approximate $y_i$ with a small error for all $i=1,2,\cdots,m$. Second, the space of $\cF$ should not be too large or complex, so that $\cF$ has favorable generalization power to generalize its performance from the training set $\cS$ to the true probability distribution that $\cS$ is sampled from. This is a reflection of the Occam's razor principle \cite{Shalev-Shwartz2014Understanding}. Therefore, we need to design a variational quantum circuit class $\cF$, which meets the following two requirements simultaneously in order to learn $f_C$:

\begin{itemize}
\item \textbf{R1}: For any $C \in \cC$, there exists a hypothesis function $f_h \in \cF$, such that $f_h(x) = f_C(x)$ for all $x \in \cX$.
\item \textbf{R2}: The hypothesis space $\cF$ satisfies  the PAC learnablity.
\end{itemize}

We note that the first requirement \textbf{R1} is stronger than the first assumption, because the function $f_h$ in \textbf{R1} is the same as $f_C$. Hence, the training error $f_h$ on the set of samples $\cS$ is necessarily zero, whereas the first assumption only requires that $f_h$ has a small training error for $\cS$. In supervised learning, obtaining high-quality training samples is usually resource-demanding in practice. Thus, studying the sample complexity becomes crucial. In the following, we will study the sample complexity of learning parametric quantum circuits and rigorously prove that any physical quantum circuit is PAC learnable.

\subsection{A family of universal variational circuits}
To meet \textbf{R1}, we should ensure that  $F$ has representation power for all the $n$-qubit parametric quantum circuits $C$ in $\cC$. 
We observe that any quantum circuit $C$ can be decomposed as a sequence of $O(n^c)$ number of $H$ gates, $R_x(\theta)$ gates, and $\CNOT$ gates (see the Proposition \ref{expressC} in Appendix A). Thus in our construction of $F$, we also use these three kinds of gates and arrange them into a uniform pattern, so that its scaling to quantum circuits with more qubits is clear. The construction of $F$ is illustrated in Fig. \ref{VariationF},  
which is based on block assembling, i.e., assembling some relatively small gadgets to form a more complicated block. By convention, we denote the three Pauli matrices by $X,Y$, and $Z$. A well-known result about quantum circuit states that any single-qubit unitary gate can be expressed as $e^{i\alpha}R_z(\beta)R_x(\gamma)R_z(\delta)$, where $\alpha, \beta,\gamma$, and $\delta$ are four real numbers, and $R_z(\theta)=e^{-i\theta Z/2}$ and $R_x(\theta)=e^{-i\theta X/2}$ are the rotation operators along the z-axis and x-axis on the Bloch Sphere,  respectively  \cite{nielsen2002quantum}. Inspired by this, we define a set of basic gadgets (we call them level-$1$ blocks) to be $L = \{HR_x(\beta)HR_x(\gamma)HR_x(\delta)H$ $\arrowvert$  $\beta,\gamma,\delta\in[0,2\pi)\}$. In this way, we can tune the parameters of a level-$1$ block so that it can represent all the single-qubit unitary gates  up to an irrelevant global phase factor   $e^{i\alpha}$. In addition, when we set $\beta=\gamma=\delta=0$, the level-$1$ block will reduce to the  identity gate.

Using level-$1$ blocks we can construct level-$2$ blocks $B_{i}$ $(i=2,3,\cdots,n)$. First we put two level-$1$ blocks (denoted as $L$) at qubit $1$ and qubit $i$, and then insert two $\CNOT_{1i}$ gates, as shown in Fig. \ref{VariationF} (b). Here, we use $\CNOT_{1i}$ to denote the controlled-NOT gate between the first and $i$-th qubits, with the first qubit being the control qubit and the $i$-th one being the target one. With this, the desired hypothesis quantum circuit $F$ can be constructed as:
\begin{eqnarray}
F=(B_n B_{n-1}\cdots B_2)^{M n^c},
\end{eqnarray}
where $M$ is a constant independent of the number of qubits $n$. We mention that the hypothesis circuit $F$ has a uniform structure for arbitrary system sizes. In addition, only the x-rotations contain variational parameters, and it is straightforward to obtain that the total number of parameters used for defining $F$ scales as $O(n^{c+1})$.

For an arbitrary quantum circuit $U$, we say that it can be represented by the variational circuit $F$ if there exists a solution to the parameters (denoted collectively as $\boldsymbol{\theta}$) in $F$   such that $F=U$ up to an irrelevant global phase. Now, we are ready to give the first theorem stating that for any $C \in \cC$, we can use $F$ to represent it:

\newcommand{\theorema}{
For any $C \in \cC$, there exists a hypothesis function $f_h \in \cF$, such that $f_h(x) = f_C(x)$ for all $x \in \cX$.
}
\begin{theorem}
\label{theorema}
\theorema
\end{theorem} 
\begin{proof}
We first give a high-level intuition for the proof.  We note that any gate acting on a constant number of qubits can be decomposed into a quantum circuit with a constant number of elementary gates, namely the CNOT gates, $H$ gates, and $R_x(\theta)$ gates. Thus, any $C \in \cC$ can be decomposed into a quantum circuit with $O(n^c)$ elementary gates. In addition, as our hypothesis quantum circuit $F$ consists of $Mn^c$ layers and every two layers can represent one arbitrary elementary gate acting on any pair of qubits through tuning the parameters properly, we can prove that $C$ can be represented by $F$ in an exact fashion.

The complete proof is as follows. First, we note that $C$ consists of $O(n^c)$ elementary gates by the definition of $\cC$.
By Proposition \ref{expressC} in Appendix A, the circuit $C$ can be written as a product form $U_l U_{l-1} \cdots U_1$, where $l = O(n^c)$, and each $U_i$ is either $\CNOT_{1j}$ or a single-qubit unitary gate $V_j$ on qubit $j$.

Denoting $B_nB_{n-1}\cdots B_2$ as one layer of level-$2$ blocks,  then we define a block string of $d$ layers as follows:
\begin{eqnarray}
(B_nB_{n-1}\cdots B_2)^d. \nonumber
\end{eqnarray}
We define $T$ as the minimal number so that a block string of  $T$ layers can represent $U_l U_{l-1} \cdots U_1$. As our hypothesis circuit $F$ contains $Mn^c$ layers in total, we will show that to represent $U_l U_{l-1} \cdots U_1$, the minimal number of layers needed in the block string is no greater than $Mn^c$. Then, when the previous $T$ layers in $F$ have represented $C$ exactly, by the first part of Proposition \ref{blockrep} in Appendix A, we can set the remaining $(Mn^c - T)$ layers to be the identity gate on the $n$ qubits. Therefore, we can show that $F$ can represent $C$ in an exact fashion and complete the proof of Theorem \ref{theorema}.

To show that $T \leq Mn^c$, we will prove that for each $U_i$, $i=1,2,\cdots,l$, we only need two layers to represent $U_i$ exactly. Then putting them together, we can prove that we need only $2l$ layers to represent $U_l U_{l-1} \cdots U_1$, which yields $T \leq 2l$.

We fix any $i \in \{1,2,\cdots,l\}$. By Proposition \ref{blockrep}, one layer $B_nB_{n-1}\cdots B_2$ can represent any single-qubit unitary gate $V_j$ acting on any qubit $j = 1,2,\cdots,n$ up to a global phase $e^{i\alpha}$. Moreover, two layers $(B_nB_{n-1}\cdots B_2)^2$ can represent $\CNOT_{1j}$ for any $j = 2,3,\cdots,n$. We note that $U_{i}$ is either a single-qubit gate $V_j$ acting on some qubit $j$, or a two-qubit gate $\CNOT_{1j}$. Therefore, we only need two layers at most to represent $U_{i}$ exactly.

As we have shown that $T \leq 2l = O(n^c)$, by choosing a large enough constant $M$, we can prove that $T \leq Mn^c$ and complete the proof of the theorem.
\end{proof}

Theorem \ref{theorema} shows that given any quantum circuit $C \in \cC$, there always exists a solution to the parameters, such that  the circuit $F$ can simulate the quantum circuit $C$ acting on the $n$-qubits with zero error. Therefore, given any training set $\cS = (x_i,y_i)_{i=1,2,\cdots,m}$ sampled independently from some distribution $P$ over $\cX \times \cY$, if we have $y_i = \fc(x_i)$ for all $i=1,2,\cdots,m$, we can find an instance $f_h \in \cF$ with zero training error. In fact, this theorem has a wider range of applications. 
When a quantum circuit consists of fewer than $n^c$ gates, we can add some identity gates after it, so our theorem covers all the quantum circuits containing no more than $n^c$ gates. In other words, we can use only $O(n^{c+1})$ gates, which are arranged in a uniform pattern, to represent all the circuits with $n^c$ gates or fewer. We remark that the number of  gates in many famous quantum circuits,  such as quantum support vector machine \cite{Rebentrost2014Quantum}, HHL algorithm \cite{Harrow2009Quantum}, and quantum Fourier transform \cite{coppersmith2002approximate}, scales polynomially with the number of qubits. Therefore, all of these circuits above can be represented exactly by our circuit $F$.

\subsection{PAC learnability of  \texorpdfstring{$\cF$}{}}
In the PAC setting, we usually assume that the input samples are randomly generated from certain unknown probability distribution. As a result, when the hypothesis space covers the underlying distribution and the training dataset is large enough, both the training and generalization error should be small. In this paper, our hypothesis space $\cF$ has been proved in Theorem \ref{theorema} to be able to cover all the parametric quantum circuits in $\cC$. Now, we study the sample complexity for training a circuit $F\in \cF$ to represent $C\in \cC$. To this end, we define a measure of the distance between two pure states in $\cY$, which is used as the loss function $\mathcal{L} : \cY \times \cY \rightarrow [0,1]$. Specifically, we define the loss function $\mathcal{L}(y_1,y_2)$ to be the trace distance of two quantum states $\ket{\psi(y_1)}$ and $\ket{\psi(y_2)}$ :
\begin{eqnarray}
\mathcal{L}(y_1,y_2) = \frac{1}{2} \| \ket{\psi(y_1)}\bra{\psi(y_1)} - \ket{\psi(y_2)}\bra{\psi(y_2)} \|_1, \nonumber
\end{eqnarray}
where $\| \cdot \|_1$ denotes the trace norm of a matrix. 
Given a hypothesis function $f_h \in \cF$ and the training set $\cS = (x_i,y_i)_{i=1,2,\cdots,m}$  sampled independently from some distribution $P$, we can define the empirical risk of $f_h$, which is also known as in-sample error:
\begin{eqnarray}
\rhat(f_h) = \frac{1}{m} \sum\limits_{i=1}^m \mathcal{L}(y_i, f_h(x_i)). \nonumber
\end{eqnarray}

The risk of a hypothesis function $f_h \in \cF$ is then defined as the average loss of $f_h$ over the probability distribution $P$:
\begin{eqnarray}
R(f_h)=\int_{\cX \times \cY} \mathcal{L}(y, f_h(x)) d P(x, y). \nonumber
\end{eqnarray}
Our goal is to find a hypothesis $f_h \in \cF$ and minimize its risk $R(f_h)$. As the parametric quantum circuit $C$ is a black box in our setting, we do not know the probability distribution $P$. In the learning process, we use the training set $\cS$ and find an empirical risk minimizer $\erm \in \cF$ over $\cS$. 
For convenience, given the training set $\cS$ and the probability distribution $P$, we define $\erm = \arg\min\limits_{h' \in \cF} \rhat(h')$ to be the empirical risk minimizer, and $h = \arg\min\limits_{h' \in \cF} R(h')$ to be the risk minimizer.  Now we formally introduce the definition of the PAC learnability for completeness \cite{wolf2018mathematical}:

\begin{definition}(PAC learnability) 
A hypothesis space $\cF$ is PAC learnable, if there exists a 
function $\nu:(0,1)^2 \rightarrow \mathbb{N}$, such that for all $(\epsilon, \delta) \in(0,1)^{2}$ and all probability measures P over $\cX \times \cY$, when the size of training set $|\cS| \geq \nu(\epsilon,\delta)$, we have
\begin{eqnarray}
\mathbb{P}_{\cS}\left(  R(\erm) - R(h)   \leq \epsilon\right) \geq 1- \delta.
\end{eqnarray}
Here $\mathbb{P}_{\cS}(A)$ denotes the probability that event $A$ happens over repeated sampling of the training set $\cS$.

\end{definition}

We note that when we randomly select a state $x \in \cX$, input it into the circuit $C$ and get an output state $y = \fc(x) \in \cY$, the resulting probability distribution $P$ of state pairs $(x,y)$ will satisfy $R(h) = 0$, because we can find an instance $f_h \in \cF$ equal to $\fc$ by Theorem \ref{theorema}. Therefore, if our $\cF$  is PAC learnable, after we prepare the training samples $\cS$ and get an empirical minimizer $\erm$, with probability $1-\delta$, the average loss of $\erm$ over $P$  will be no larger than $\epsilon$, i.e., $R(\erm) \leq \epsilon$. Now we are going to prove that $\cF$ is PAC learnable, and the sample complexity $\nu$ is polynomial with $n$, $1/\epsilon$, and $\ln\frac{1}{\delta}$.

\newcommand{\theoremb}{
The hypothesis space $\cF$ satisfies the PAC learnablity, with sample complexity $\nu(\epsilon,\delta) = O(\frac{1}{\epsilon^2}(n^{c+1}\ln\frac{n}{\epsilon}+\ln\frac{1}{\delta}))$.
}
\begin{theorem}
\label{theoremb}
\theoremb
\end{theorem}
\begin{proof}
The essential idea for the proof relies on the discretization of $\cF$. First, we construct a finite set of hypothesis functions $\cF'$, such that  for each function  $\fh \in \cF$, we can find a function $\fh' \in \cF'$  close enough to $\fh$. Then we use Lemma \ref{finitef} in Appendix B to show that $\cF'$ is PAC learnable. Finally, for any $\fh \in \cF$ and its corresponding $\fh'$, as $\fh$ and $\fh'$ are close enough, we can prove that their risk and empirical risk are close as well. Therefore, we obtain that  $\cF$ is PAC learnable.

For clarity, we denote $l$ as the total number of $R_x(\theta)$ gates in the circuit $F$, and observe that $l \leq 12Mn^{c+1}$. We recall that $\cF$ is defined as the collection of all the functions $f_h: \cX \rightarrow \cY$ that the circuit $F$ can represent by tuning the value of the parameters $\boldsymbol{\theta} = (\theta_1,\theta_2,\cdots,\theta_l)$, where $\theta_i \in [0,2\pi)$ is the variational parameter characterizing the $i$-th x-rotation. 
Now we define a finite set $\cF' \subseteq \cF$ in this way: $\cF'$ is the collection of all the functions $f_h': \cX \rightarrow \cY$ that circuit $F$ can represent by tuning the value of all the $\theta_i$ in $\{0, e, 2e, \cdots, Ne \}$, where $e = \frac{\epsilon}{6 K n^{c+1}}$, $N = \lfloor \frac{2\pi}{e} \rfloor$,  and $K$ is a large enough constant. As there are  $l=O(n^{c+1})$ rotational gates in circuit $F$ in total, we have \begin{eqnarray}
|\cF'| = \left\lceil\frac{12 \pi K n^{c+1}}{\epsilon}\right\rceil ^ {O(n^{c+1})}, \nonumber \end{eqnarray}
which is finite. As a result, we can plug $\cF'$ and $\epsilon'=\epsilon/6$ into Lemma \ref{finitef} to obtain that when $|\cS| \geq \frac{18}{\epsilon^2}(\ln |\cF'| + \ln\frac{2}{\delta})$, 
\begin{eqnarray}
\mathbb{P}_{\cS}\left(\forall \fh^{\prime} \in \cF^{\prime}:\left|\hat{R}\left(\fh^{\prime}\right)-R\left(\fh^{\prime}\right)\right| \leq \frac{\epsilon}{6}\right) \geq 1-\delta. \label{ff1}
\end{eqnarray}

We fix all the parameters $\boldsymbol{\theta}$ in circuit $F$, and we will get an arbitrary hypothesis function $f_h \in \cF$. Then we can round all the parameter $\boldsymbol{\theta}$ of circuit $F$ into their nearest multiples of $e$ in $\{0,e,2e,\cdots, Ne\}$, and we will get a new hypothesis function $\tfh \in \cF'$. By Proposition \ref{near}, we obtain that for any $\fh \in \cF$,
\begin{eqnarray}
\left|R(\tfh) - R(\fh) \right| & \leq & \frac{\epsilon}{6}, \\  \label{ff2}
\left|\rhat(\tfh) - \rhat(\fh) \right| & \leq & \frac{\epsilon}{6}.  \label{ff3}
\end{eqnarray}
Combining the three inequalities (\ref{ff1}-\ref{ff3}) together, we arrive at
\begin{eqnarray}
\mathbb{P}_{\cS}\left(\forall \fh \in \cF:\left|\hat{R}\left(\fh\right)-R\left(\fh\right)\right| \leq \frac{\epsilon}{2}\right) \geq 1-\delta, \label{ff4}
\end{eqnarray}
when $|\cS| \geq \frac{18}{\epsilon^2}(\ln |\cF'| + \ln\frac{2}{\delta})$. To prove that $\cF$ is PAC learnable, we recall our notations that $h = \arg\min\limits_{h' \in \cF} R(h')$, and $\erm = \arg\min\limits_{h' \in \cF} \rhat(h')$. Combining the inequality (\ref{ff4}) and Proposition \ref{unifconv}, we obtain that
\begin{eqnarray}
\mathbb{P}_{\cS}\left( R(\erm) - R(h) \leq \epsilon \right) \geq 1-\delta, \nonumber
\end{eqnarray}
when $|\cS| \geq \frac{18}{\epsilon^2}(\ln |\cF'| + \ln\frac{2}{\delta})$. Plugging in  $|\cF'| = \left\lceil\frac{12 \pi K n^{c+1}}{\epsilon}\right\rceil ^ {O(n^{c+1})}$, we can prove that the hypothesis space $\cF$ is PAC learnable with sample complexity \begin{eqnarray}
\nu (\epsilon, \delta) = O(\frac{1}{\epsilon^2}(n^{c+1}\ln\frac{n}{\epsilon}+\ln\frac{1}{\delta})). \nonumber
\end{eqnarray}
This completes the proof of Theorem \ref{theoremb}.
\end{proof}

\newcommand{\cE}{\mathcal{E}}
We denote $\cE$ as the collection of all the functions $f_C:\cX\rightarrow\cY$, where $C \in \cC$. In fact, we can prove that $\cE$ is PAC learnable as well. By Theorem \ref{theorema}, our hypothesis space $\cF$ can cover all the quantum circuits in $\cC$. Thus we can obtain that $\cE \subseteq \cF$. Using the inequality (\ref{ff4}), we will arrive at
\begin{eqnarray}
\mathbb{P}_{\cS}\left(\forall \fh \in \cE:\left|\hat{R}\left(\fh\right)-R\left(\fh\right)\right| \leq \frac{\epsilon}{2}\right) \geq 1-\delta, \label{ff5}
\end{eqnarray}
when $|\cS| \geq \frac{18}{\epsilon^2}(\ln |\cF'| + \ln\frac{2}{\delta})$. Similarly with the method in Theorem $\ref{theoremb}$, we combine the inequality (\ref{ff5}) with Proposition $\ref{unifconv}$, and we can prove that $\cE$ is PAC learnable with sample complexity
$\nu (\epsilon, \delta) = O(\frac{1}{\epsilon^2}(n^{c+1}\ln\frac{n}{\epsilon}+\ln\frac{1}{\delta}))$ as well.

We stress the differences between our results and the previous works \cite{Chung2021Sample,bu2021statistical} in the literature. First, in Ref. \cite{Chung2021Sample} Chung and Lin focused on a \textit{finite} set of discretized quantum channels, and their algorithm is based on random orthogonal measurements. Whereas, in our settings we focus on a set of unitary quantum circuits with continuous variational parameters, thus the size of our concept class is \textit{infinite}. Moreover, our proof is based on a family of variational quantum neural networks with an explicit uniform structure, which would be useful in practical applications. Second, in Ref. \cite{bu2021statistical} Bu \textit{et al}. considered a more general class of quantum channels and their bounds of the sample complexity grow exponentially with the number of qubits $n$. In contrast, our focus here is variational quantum circuits and the sample complexity we obtained scales only polynomially with the system size. In other words, while Ref. \cite{bu2021statistical}'s setting is more general, the sample complexity bounds obtained in this work is exponentially tighter. Our work and Refs. \cite{Chung2021Sample,bu2021statistical} are complementary to each other.

It is also worthwhile to clarify that, although we have proved that the sample complexity for learning any physical quantum circuit is low (namely, it only scales polynomially with the number of qubits  involved), this does not mean that these circuits can be learned efficiently since the time complexity to learn an unknown circuit can still be exponentially high. In fact, it has been proved recently that training a variational quantum circuit, even for logarithmically many qubits and free fermionic systems, is NP-hard \cite{bittel2021training}. This implies that although we know for sure that our hypothesis $\cF$ can cover all physical quantum circuits and only a polynomial number of samples are needed to train a variational circuit $F\in\cF$, how to efficiently solve the optimization problem of minimizing the empirical risk remains unclear and might be an exponentially hard problem in practice.

\section{Discussion}
We mention that the family of hypothesis quantum variational circuits constructed in this paper is of independent interest due to its use of only $O(n^{c+1})$ variational parameters while maintaining notable representation power. These circuits might be used as variational ansatz for implementing quantum classifiers \cite{Schuld2020Circuitcentric,Farhi2018Classification,Havlicek2019Supervised,Zhu2019Training,Cong2019Quantum,Grant2018Hierarchical,Lu2020Quantum}, variational quantum eigensolvers \cite{Peruzzo2014Avariational,Kokail2019Self,Liu2019Variational,Wang2019Accelerated}, or quantum generative adversarial networks \cite{Lloyd2018Quantum,Demers2018Quantum,hu2019quantum}, etc. On the other hand, we also remark that similarly to many other variational quantum circuits constructed in the literature,  this family of variational circuits may suffer from the barren plateau (i.e., vanishing gradient) problem \cite{Mcclean2018Barren,Cerezo2020Cost} as well. In addition, our work can be appealing as the family of circuits is constructed without optimizing the structure and the number of parameters. In the future, it would also be interesting to explore other alternative structures with smaller depths and fewer parameters. Another interesting problem worth further investigation is to consider a scenario where we do not have perfect  knowledge about the training data, namely that the training dataset may not be fully labelled. How to extend our results to this scenario remains unknown.

We note that in our proof, the use of PAC learning theory is in fact independent from the learning model, i.e., it can deal with both the classical and quantum objectives. In our settings, the objects to be learned are parametric quantum circuits, but we can still use standard classical techniques of PAC learning theory (like discretization) to obtain the sample complexity bound.

In summary, we have proved that  unitary physical quantum circuits are PAC learnable on a quantum computer via empirical risk minimization. In particular, we proved that to learn a unitary quantum circuit with at most $n^c$ local gates, the sample complexity is bounded by $\tilde{O}(n^{c+1})$. Our results are generally applicable to all unitary quantum circuits of practical interest. There are many notable quantum circuits (algorithms or kernels, such as Shor's factorization algorithm \cite{shor1994algorithms}, the HHL algorithm \cite{Harrow2009Quantum}, quantum support vector machine \cite{Rebentrost2014Quantum}, quantum classification based on discrete logarithm  \cite{liu2021rigorous}, etc.) that hold the intriguing potential of exponential quantum speedup. Our results imply that a polynomial number of samples are enough to learn these quantum circuits. In Ref. \cite{bang2014strategy}, Bang \textit{et al}. proposed a method for learning quantum algorithms assisted by machine learning, which shows learning speedup in designing quantum circuits for solving the Deutsch–Jozsa problem, and our results imply that the quantum circuits they used are PAC learnable as well. 

\section{Acknowledgments}

We thank Wenjie Jiang, Peixin Shen, and Xun Gao in particular for their helpful discussions. This work is supported by the start-up fund from Tsinghua University (Grant. No. 53330300320), the National Natural Science Foundation of China (Grant. No. 12075128), and the Shanghai Qi Zhi Institute.

\section*{References}
\bibliographystyle{iopart-num}
\bibliography{DengQAIGroup}

\newpage
\appendix

\section{The universality of \texorpdfstring{$\mathcal{F}$}{}}

In this paper, all the constants such as $b$, $c$, $K$, and $M$ are independent of $n$, $\epsilon$, and $\delta$. Also, we recall that $\cC$ is the set of all the $n$-qubit quantum circuits with at most $n^c$ unitary gates, with each gate acting on at most $b$ qubits.

In proving Theorem \ref{theorema} in the main text, we used three lemmas, which are appended in the following. The Lemma \ref{universal} is proved in Ref. \cite{nielsen2002quantum}, which we recap here for completeness. The Proposition \ref{expressC} and Proposition \ref{blockrep} are proved in this paper.

\newcommand{\universal}{
An arbitrary unitary operation on $b$ qubits can be implemented using a circuit containing at most $c_0 b^2 4^b$ single-qubit unitary gates and $\CNOT$ gates, where $c_0$ is a constant.
}
\begin{lemma}[\cite{nielsen2002quantum}, Section 4.5.2]
\label{universal}
\universal
\end{lemma}

\newcommand{\expressC}{
\bk{
For any $C \in \cC$, there exist $l = O(n^c)$ unitary gates $U_1,U_2,\cdots, U_l$, such that $C = U_l U_{l-1} \cdots U_1$, and each gate $U_i$ is either a single-qubit unitary gate $V_j$ acting on qubit $j$, or $\CNOT_{1j}$ gate with the first qubit being the control qubit and the $j$-th qubit being the target one.}
}
\begin{proposition}
\label{expressC}
\expressC
\end{proposition}

\begin{proof}
 We first prove that $C$ can be decomposed as $O(n^c)$ elementary gates, including CNOT gates and single-qubit unitary gates. By Lemma \ref{universal}, $C$ can be implemented by at most $c_0b^24^bn^c = O(n^c)$ unitary gates, and each gate is either a single-qubit unitary gate or  $\CNOT_{ij}$ gate with the control qubit $i$ and the target qubit $j$.

To prove that $C$ can be decomposed as the product of $l = O(n^c)$ single-qubit unitary gates and $\CNOT_{1j}$ gates, we need only prove that when $i\neq 1$ and $i\neq j$, $\CNOT_{ij}$ can be decomposed as $\CNOT_{1i}$, $\CNOT_{1j}$,  and $H$ gates.

When $i>1$, $j=1$, we can write $\CNOT_{i1}$ in this way:
\begin{eqnarray}
\CNOT_{i1} = (H_1 \otimes H_i) \CNOT_{1i} (H_1 \otimes H_i). \nonumber
\end{eqnarray}

Meanwhile, when $i,j > 1$ and $ i \neq j$,  we can decompose $\CNOT_{ij}$ into $\CNOT_{1j}$ and $\CNOT_{i1}$ in this way:
\begin{eqnarray}
\CNOT_{i j} = ({\CNOT}_{1 j}\CNOT_{i 1})^2, \nonumber
\end{eqnarray}
and we have shown that $\CNOT_{i1}$ can be decomposed as $\CNOT_{1i}$ and $H$ gates.

As each decomposition uses only $O(1)$ gates, we can obtain that $C$ can be decomposed as the product of $l = O(n^c)$  single-qubit unitary gates and $\CNOT_{1j}$ gates, and the proof is completed.
\end{proof}

In a level-$2$ block $B_j$, there are two level-$1$ blocks on qubit $1$ and qubit $j$, respectively. Each level-$1$ block has three parameters $\beta, \gamma, \delta$, and by $Z$-$X$ decomposition \cite{nielsen2002quantum}, we can tune these three parameters to enable a level-$1$ block $L_j$ on qubit $j$ to represent any single-qubit unitary gate acting on qubit $j$ up to an irrelevant global phase. Also, by setting the three parameters to zero, a level-$1$ block $L_j$ can also represent the identity gate. We will prove that by tuning the parameters of the level-$2$ blocks, one layer consisting of $\layer$ can represent any single-qubit unitary gate acting on any qubit $j$, and  $\CNOT_{1j}$ can be represented by two layers.

\newcommand{\blockrep}{
\bk{
1. One layer $\layer$ can represent any single-qubit unitary gate $V_j$ acting on any qubit $j$ up to an irrelevant global phase.}

\bk{
2. Two layers $(\layer)^2$ can represent $\CNOT_{1j}$ up to an irrelevant global phase.}
}
\begin{proposition}
\label{blockrep}
\blockrep
\end{proposition}

\begin{proof}
To prove this lemma, we will set most of the level-$2$ blocks in the layers to be the identity gates and use at most two blocks to represent the gates we need. 

We prove part one first. We separate the claim into two cases, $j = 1$ and $j \neq 1$.
When $j = 1$, we can let $B_{n}B_{n-1}\cdots B_3$ represent the identity gate by tuning all their parameters to zero. For clarity, we denote $L_j$ as a level-1 block acting on the $j$-th qubit. Given any unitary gate $V_1$ on qubit $1$, a level-$1$ block $L_1$ can represent $V_1$, and both level-$1$ blocks $L_1$ and $L_2$ can represent the identity gate. As a level-$2$ block $B_2$ consists of four level-$1$ blocks and two CNOT gates, we can tune the parameters of the four level-$1$ blocks in the following way so that $B_2$ can represent $V_1$:

\begin{figure}[H]
\centerline{
\Qcircuit @C=1em @R=.7em{
\lstick{}&\gate{V_1}&\ctrl{1}&\gate{I_1}&\ctrl{1}&\qw\\
\lstick{}&\gate{I_2}&\targ   &\gate{I_2}&\targ   &\qw}
}
\end{figure}

Similarly, when $j \neq 1$, we can let $B_nB_{n-1}\cdots B_{j+1}$ and $B_{j-1}\cdots B_3 B_2$ represent the identity gate. Then we need only let $B_j$ represent unitary gate $V_j$.   Given any unitary gate $V_j$ on qubit $j$, we can tune the parameters of the four level-$1$ blocks in the following way so that $B_j$ can represent $V_j$:
\begin{figure}[H]
\centerline{
\Qcircuit @C=1em @R=.7em{
\lstick{}&\gate{I_1}&\ctrl{1}&\gate{I_1}&\ctrl{1}&\qw\\
\lstick{}&\gate{V_j}&\targ   &\gate{I_j}&\targ   &\qw}
}
\end{figure}

Therefore, the proof of part one is completed. Now we will prove part two. We set all the parameters in the two layers $(\layer)^2$ to be zero except the two $B_j$ blocks. Then we will use two $B_j$ blocks to represent $\CNOT_{1j}$. We decompose $\CNOT_{1j}$ up to an irrelevant global phase factor $e^{-i\pi/4}$ in the following way:
\begin{eqnarray}
\left(W_4 \otimes W_{3}\right)\CNOT_{1j}\left(I_{1} \otimes W_{2}\right)\CNOT_{1j}\left(I_{1} \otimes W_{1}\right), \nonumber
\end{eqnarray}
where we set  $W_{1}=R_{z}\left(\frac{\pi}{2}\right), W_{2}=R_{y}\left(\frac{\pi}{2}\right),  W_{3}=R_{z}\left(-\frac{\pi}{2}\right) R_{y}\left(-\frac{\pi}{2}\right)$, and $W_4 =  R_{z}\left(-\frac{\pi}{2}\right)$. Here we denote  $R_z(\theta)=e^{-i\theta Z/2}$ and $R_y(\theta)=e^{-i\theta Y/2}$ as the rotation operators along the z-axis and y-axis on the Bloch Sphere, respectively. In addition, $W_4$ and the identity gate $I_1$ act on the first qubit, and $W_1, W_2$, and $W_3$ act on the $j$-th qubit.

Hence, we use the $B_j$ block in the first layer to represent  $\CNOT_{1j}\left(I_{1} \otimes W_{2}\right)\CNOT_{1j}\left(I_{1} \otimes W_{1}\right)$ in this way:
\begin{figure}[H]
\centerline{
\Qcircuit @C=1em @R=.7em{
\lstick{}&\gate{I_1}&\ctrl{1}&\gate{I_1}&\ctrl{1}&\qw\\
\lstick{}&\gate{W_1}&\targ   &\gate{W_2}&\targ&\qw}
}
\end{figure}
Finally, we use the second level-$2$ block $B_j$ to represent $W_4 \otimes W_3$, where $W_4$ acts on the first qubit and $W_3$ acts on the $j$-th qubit. Therefore, two layers $(\layer)^2$ can represent $\CNOT_{1j}$ up to an irrelevant global phase, and this completes the proof of part two.
\end{proof}

\section{PAC learnability of \texorpdfstring{$\mathcal{F}$}{}}

The following lemma shows that any finite hypothesis space $\cF'$ is PAC learnable.

\newcommand{\finitef}{Assume that the hypothesis space $\mathcal{F}'$ is finite,  
 $\delta \in(0,1]$, $\epsilon>0$ and the range of the loss function is in an interval of length $c \geq 0$. Then if the size of the training set $
|S| \geq \frac{c^{2}}{2 \epsilon^{2}}\left(\ln |\mathcal{F}'|+\ln \frac{2}{\delta}\right)
$, the event
 $\forall \fh \in \mathcal{F}':|\hat{R}( \fh)-R( \fh)| \leq \epsilon$ holds with probability at least $1-\delta$ over repeated sampling of the training set $S$.
}
\begin{lemma}[\cite{wolf2018mathematical}, Corollary 1.2]
\label{finitef}
\finitef
\end{lemma}
\newcommand{\fone}{\theta^{F_1}}
\newcommand{\ftwo}{\theta^{F_2}}
\newcommand{\fa}{f_1}
\newcommand{\fb}{f_2}

Our circuit $F$ consists of $R_x(\theta)$, $H$, and $\CNOT$ gates. By assigning two sets of different values to the variational parameters $\boldsymbol{\theta}$, we can get two distinct circuits $F_1$ and $F_2$, and their corresponding hypothesis functions $f_1$ and $f_2$ are different. We note that although $F_1$ and $F_2$ differ in the value of their variational parameters $\boldsymbol{\theta}$, their ordering of the gates ($R_x(\theta)$, $H$, and $\CNOT$ gates) are the same.
We will show that when all the variational parameters in circuit $F_1$ and $F_2$ are close enough, the risk and empirical risk of $f_1$ and $f_2$ will be close. To prove this, we first define the distance of two unitary matrices $U_1$, $U_2 \in \mathbb{C}^{2^n \times 2^n}$ as the $2$-norm of the matrix $U_1 - U_2$: 
\begin{eqnarray}
E(U_1,U_2) = \|U_1 - U_2\|_2 =  \sup\limits_{x \in \cX} \|(U_1 - U_2)\ket{\psi(x)}\|_2. \nonumber
\end{eqnarray}
Now we introduce the following lemma about the function $E(U_1,U_2)$. 

\newcommand{\funcE}{
\bk{
   The function $E(U_1,U_2)$ satisfies the following properties:}
   
\bk{   1. Let $U$, $V$ be the $R_x(\theta)$, $R_x(\theta+\epsilon)$ gates acting on the $j$-th qubit, respectively, where
$ \epsilon \in(0,1), \theta \in[0,2 \pi), j = 1,2,\cdots, n$. Then $E(U,V) \leq \epsilon$.}
   
\bk{   2. $E\left(U_l U_{l-1} \cdots U_1, V_{l} V_{l-1} \cdots V_{1}\right) \leq \sum\limits_{j=1}^{l} E\left(U_{j}, V_{j}\right),$
   where $U_1,U_2,\cdots,U_l,V_1,V_2,\cdots,V_l$ are unitary matrices. }
}
\begin{proposition}
\label{funcE}
\funcE
\end{proposition}

\begin{proof}
The second property is shown in \cite{nielsen2002quantum}, Section 4.5.3. We need only prove the first property.
\begin{eqnarray}
E(U,V) &= \|U-V\|_2 \nonumber \\ 
&= \|R_x(\theta) - R_x(\theta + \epsilon)\|_2 \nonumber \\
&= \left\|I-R_{x}(\epsilon)\right\|_2 \nonumber\\
&\myeqi \|I-(I-i \epsilon X / 2+ (i \epsilon X / 2)^2/(2!) - \cdots)\|_2 \nonumber\\
&\leq \|i \epsilon X / 2\|_2 + \|(i \epsilon X / 2)^2/(2!)\|_2 + \cdots \nonumber \\
&\myeqii \exp(\epsilon/2)-1 \nonumber \\
&\leq \epsilon, \nonumber
\end{eqnarray}
where (i) uses Taylor's expansion of the operator $R_{x}(\epsilon) = e^{-i \epsilon X / 2}$, and (ii) uses Taylor's expansion of $\exp(\epsilon/2)$ and that $\|X\|_2 = \|I\|_2 =  1$. 
\end{proof}

We recall that $\cL(y_1,y_2)$ is the trace distance of two pure states $\ket{\psi(y_1)}$ and $\ket{\psi(y_2)}$.  Then we introduce the following properties of $\cL(y_1,y_2)$.

\newcommand{\ya}{\ket{\psi(y_1)}}
\newcommand{\ay}{\bra{\psi(y_1)}}
\newcommand{\yb}{\ket{\psi(y_2)}}
\newcommand{\yac}{\ket{\psi(y_1)^{\perp}}}
\newcommand{\yaa}{\ket{\psi(y_1)}\bra{\psi(y_1)}}
\newcommand{\ybb}{\ket{\psi(y_2)}\bra{\psi(y_2)}}
\newcommand{\ycc}{\ket{\psi(y_1)^{\perp}}\bra{\psi(y_1)^{\perp}}}
\newcommand{\yca}{\ket{\psi(y_1)^{\perp}}\bra{\psi(y_1)}}
\newcommand{\yacc}{\ket{\psi(y_1)}\bra{\psi(y_1)^{\perp}}}
\newcommand{\funcL}{
\bk{
The function $\cL: \cY \times \cY \rightarrow [0,1]$ satisfies the following two  properties:}
\bk{
\begin{enumerate}
   \item For any $y_1,y_2,y_3 \in \cY$, we have  
   $\cL(y_1,y_3) - \cL(y_2,y_3) \leq \cL(y_1,y_2)$. 
   \item For any $y_1,y_2 \in \cY$, we have 
   $\cL(y_1,y_2) \leq \|\ket{\psi(y_1)} -  \ket{\psi(y_2)} \|_2$.
\end{enumerate}}
}
\begin{proposition}
\label{funcL}
\funcL
\end{proposition}

\begin{proof}
The first part of this lemma is the  triangle inequality, which is proved in \cite{nielsen2002quantum}, Section 9.2.1. Here, we only prove the second property.
We denote $F(\ya,\yb) = \left|\langle \psi(y_1) | \psi(y_2) \rangle\right|$ as the fidelity between the two states $\ya$ and $\yb$. Then we will arrive at \begin{eqnarray}
\cL(y_1,y_2) &= \frac{1}{2}\| \yaa - \ybb \|_1  \nonumber \\
&\myeqiii \sqrt{1 - F(\ya,\yb)^2}, \nonumber
\end{eqnarray}
where the proof of equation (iii) is given in  \cite{nielsen2002quantum}, Section 9.2.3.

In addition, we note that for any complex number $z\in \mathbb{C}$ and its complex conjugate $z^* \in \mathbb{C}$, as $(|z| - 1)^2 \geq 0$, we have $2-2|z| \geq 1-|z|^2$.
Hence, we get $2-z-z^* \geq  2-2|z| \geq 1-|z|^2$.
Let $z = \langle \psi(y_1) | \psi(y_2) \rangle$, we obtain that 
\begin{eqnarray}
\cL(y_1,y_2) &= \sqrt{1 - F(\ya,\yb)^2} = \sqrt{1 - |z|^2} \nonumber \\
&\leq \sqrt{2-z-z^*} = \|\ya - \yb \|_2. \nonumber
\end{eqnarray}
This completes the proof.
\end{proof}

Now we will use the properties of $E(U_1,U_2)$ and $\cL(y_1,y_2)$ to show that the differences of both the risk and empirical risk between $f_1$ and $f_2$ are bounded by $\epsilon$, where the hypothesis functions $f_1$ and $f_2$ correspond to the variational circuits $F_1$ and $F_2$, respectively.

\newcommand{\near}{
\bk{
We denote $\boldsymbol{\theta}^{F_1} = (\fone_1, \fone_2, \cdots, \fone_l)$ as a vector containing all the variational  parameters in $F_1$, where $l$ is the number of $R_x(\theta)$ gates in circuit $F$, and $\fone_i$ is the value of the variational parameter characterizing the $i$-th x-rotation of $F_1$.
Similarly, we denote $\boldsymbol{\theta}^{F_2} = (\ftwo_1, \ftwo_2, \cdots, \ftwo_l)$ as a vector containing all the variational parameters in $F_2$.}

\bk{Let $f_1, f_2 \in \cF$ be the corresponding hypothesis functions of $F_1, F_2$, respectively.
Then given any probability distribution $P$ over $\cX \times  \cY$ and training set $S = (x_i,y_i)_{i=1,2,\cdots,m}$,  
the following two inequalities hold if $\|\boldsymbol{\theta}^{F_1} - \boldsymbol{\theta}^{F_2}\| _{\infty} \leq \frac{\epsilon}{K n^{c+1}}$ ($K$ is a large enough constant):}
\bk{\begin{eqnarray}
|R(f_1) - R(f_2)| \leq \epsilon, \nonumber \\
|\rhat(f_1) - \rhat(f_2)| \leq \epsilon. \nonumber 
\end{eqnarray}}
}
\begin{proposition}
\label{near}
\near
\end{proposition}

\begin{proof}
   First, we will prove that $E(F_1,F_2) \leq \epsilon$ when $\|\boldsymbol{\theta}^{F_1} - \boldsymbol{\theta}^{F_2}\|_{\infty} \leq \frac{\epsilon}{K n^{c+1}}$. Then we will use it to show the risk and empirical risk of $\fa$ and $\fb$ are close.
   
   As $F$ is composed of $H$ gates, $R_x(\theta)$ gates and $\CNOT$ gates, we can write $F_1 = U_l U_{l-1} \cdots U_1 $ and $F_2 = V_l V_{l-1} \cdots V_1 $, where $U_i$ is the $i$-th gate in $F_1$, and $V_i$ is the $i$-th gate in $F_2$. As $U_i$ and $V_i$ are of the same type of gates, we can prove that $E(U_i, V_i) \leq \frac{\epsilon}{K n^{c+1}}$ by separating different cases on the types of $U_i$ and $V_i$:
   
 Case I:  If $U_i$ and $V_i$ are both $H$ gates or both $\CNOT$ gates, as there is no variational parameter in $H$ or $\CNOT$, we have $U_i = V_i$, and we obtain that $E(U_i,V_i) = 0$.
   
Case II:   If $U_i$ and $V_i$ are both $R_x(\theta)$ gates, as the difference of $\theta_i^{F_1}$ and $\theta_i^{F_2}$ is at most $\frac{\epsilon}{K n^{c+1}}$, by the first property of Proposition \ref{funcE}, we have $ E(U_i,V_i) \leq \frac{\epsilon}{K n^{c+1}}$.
   
   We note that $l = O(n^{c+1})$ by our construction of $F$. By the second property of Proposition \ref{funcE} and choosing a large enough constant $K$ such that $l \leq Kn^{c+1}$, we can get 
   \begin{eqnarray}
E(F_1,F_2) &= E(U_l U_{l-1} \cdots U_1, V_l V_{l-1} \cdots V_1) \nonumber \\
&\leq \sum_{j=1}^{l} E\left(U_{j}, V_{j}\right) \nonumber \\
&\leq \sum_{j=1}^{l} \frac{\epsilon}{K n^{c+1}} \leq \epsilon. \nonumber
   \end{eqnarray}

Now we can bound the differences of the risk and empirical risk between the two hypothesis functions $\fa$ and $\fb$, respectively. For convenience, we define $D(f_1,f_2)$ as $\sup\limits_{x \in \cX, y \in \cY} \bigg| \cL(y, \fa(x)) - \cL(y, \fb(x)) \bigg|$. We observe that both $|R(f_1) - R(f_2)|$ and $|\rhat(f_1) - \rhat(f_2)|$ can be bounded by $D(f_1,f_2)$. Hence, we will prove that $D(f_1,f_2) \leq \epsilon$, and we can obtain the two inequalities $|R(f_1) - R(f_2)| \leq \epsilon$ and $|\rhat(f_1) - \rhat(f_2)| \leq \epsilon$.
\begin{eqnarray}
D(f_1,f_2) &\myineqiv \sup\limits_{x \in \cX}    \cL(\fa(x), \fb(x)) \nonumber \\
&\myineqv \sup\limits_{x \in \cX} \| \fa(x) - \fb(x) \|_2 \ \nonumber \\
&= \sup\limits_{x \in \cX} \| (F_1  - F_2) \ket{\psi(x)} \|_2 \nonumber \\
&= E(F_1,F_2) \leq \epsilon, \nonumber
\end{eqnarray}
where (iv) uses the first property of function $\cL$ in Proposition \ref{funcL}, and (v) uses the second property of function $\cL$ in Proposition \ref{funcL}. This completes the proof of Proposition \ref{near}.
\end{proof}

We note that in our proof of Theorem \ref{theoremb}, we used Lemma \ref{finitef}  and Proposition \ref{near} to show that $\forall \fh \in \cF:  \left|\hat{R}\left(\fh\right)-R\left(\fh\right)\right| \leq \frac{\epsilon}{2}$ holds with probability $1-\delta$. To prove that $\cF$ is PAC learnable, we introduce the following technical lemma.

\newcommand{\unifconv}{
\bk{
Assume  $\forall \fh \in \cF:  \left|\hat{R}\left(\fh\right)-R\left(\fh\right)\right| \leq \frac{\epsilon}{2}$ holds. We denote $h = \arg\min\limits_{h' \in \cF} R(h')$, and $\erm = \arg\min\limits_{h' \in \cF} \rhat(h')$. Then we have
\begin{eqnarray}
R(\erm) - R(h) \leq 2\sup\limits_{h' \in \cF} \left|\rhat(h') - R(h')\right| \leq \epsilon. \nonumber
\end{eqnarray}
}
}
\begin{proposition}
\label{unifconv}
\unifconv
\end{proposition}

\begin{proof}
The proof of this inequality is given in \cite{wolf2018mathematical}, Section 1.2. We give the proof of the lemma here for completeness. To bound $R(\erm) - R(h)$, we    observe that it can be expressed as the sum of $R(\erm) - \rhat(\erm)$ and $\rhat(\erm) - R(h)$. Then we can use $\sup\limits_{h' \in \cF} \left|\rhat(h') - R(h')\right|$ to bound $R(\erm) - \rhat(\erm)$ and $\rhat(\erm) - R(h)$, respectively.
\begin{eqnarray}
R(\erm) - R(h) 
&=& R(\erm) - \rhat(\erm) + \rhat(\erm) - R(h) \nonumber \\
&\myeqvi &  R(\erm) - \rhat(\erm) + \sup\limits_{h' \in \cF} (\rhat(\erm) - R(h')) \nonumber \\
&\myineqvii&  R(\erm) - \rhat(\erm) + \sup\limits_{h' \in \cF} (\rhat(h') - R(h')) \nonumber \\
&\leq& |\rhat(\erm) - R(\erm)| + \sup\limits_{h' \in \cF} |\rhat(h') - R(h')| \nonumber \\
&\leq& 2\sup\limits_{h' \in \cF} \left|\rhat(h') - R(h')\right| \leq \epsilon, \nonumber \end{eqnarray}
where (vi) uses that $h = \arg\min\limits_{h' \in \cF} R(h')$, and (vii) uses that $\erm = \arg\min\limits_{h' \in \cF} \rhat(h')$.
\end{proof}

\end{document}